\documentclass[12pt]{article}
 \pdfoutput=1
\usepackage{amsmath,amsthm,amsfonts}
\usepackage{fullpage}
\usepackage[skip=10pt plus1pt, indent=40pt]{parskip}
\usepackage{t1enc}
\usepackage{appendix}
\usepackage{enumerate}
\usepackage{epsfig}
\usepackage{makeidx}
\usepackage{amssymb}
\usepackage{amsfonts,amsmath,amsthm,amssymb,nccmath}
\usepackage[authoryear]{natbib}
\usepackage{mathabx,relsize}
\usepackage{graphics,graphicx,psfrag,epsf}
\usepackage{enumerate}
\usepackage{pdflscape}
\usepackage{afterpage}
\usepackage{capt-of}
\usepackage{lscape}
\usepackage{float}
\usepackage{caption}
\usepackage{subcaption}
\usepackage{comment}
\usepackage{lipsum}

\usepackage{ctable} 

\allowdisplaybreaks[3]

\usepackage{xcolor}
\definecolor{dbl}{rgb}{0.0,0.0,1}

\usepackage[pdftex,hypertexnames=false,linktocpage=true,colorlinks]{hyperref}
\hypersetup{colorlinks=true,linkcolor=dbl,anchorcolor=dbl,citecolor=dbl,filecolor=dbl,urlcolor=dbl,bookmarksnumbered=true}\RequirePackage{hypernat}

\usepackage{listings}
\usepackage{multirow}
\usepackage[utf8]{inputenc}
\usepackage[utf8]{luainputenc}
\usepackage[english]{babel}
\usepackage[top=1.5in,bottom=1.2in,left=1in,right=1in]{geometry}

\makeindex

\renewcommand\thesection{\arabic{section}}

\newcommand{\ncom}{\newcommand}
\ncom{\ul}{\underline}
\ncom{\m}{\mathbb{R}}
\ncom{\np}{N_P(\mu,\sGm)}
\ncom{\nno}{\nonumber}
\ncom{\non}{\nonumber}
\ncom{\ds}{\displaystyle}
\ncom{\half}{\frac{1}{2}}
\ncom{\mbx}{\makebox{.25cm}}
\ncom{\hs}{\mbox{\hspace{.25cm}}}
\ncom{\rar}{\rightarrow}
\ncom{\Rar}{\Rightarrow}
\ncom{\noin}{\noindent}
\ncom{\sz}{\scriptsize}
\ncom{\rf}{\ref}
\ncom{\s}{\sqrt{2}}
\ncom{\sgm}{\sigma}
\ncom{\sGm}{\Sigma}
\ncom{\Sgm}{\sigma^2}
\ncom{\psgm}{\sigma^{\prime}}
\ncom{\dt}{\delta}
\ncom{\Dt}{\Delta}
\ncom{\lmd}{\lambda}
\ncom{\Lmd}{\Lambda}
\ncom{\Th}{\theta}
\ncom{\e}{\eta}
\ncom{\Ch}{\chi^2}
\ncom{\pcc}{\stackrel{P}{>}}
\ncom{\lp}{\stackrel{L_{p}}{>}}
\ncom{\dist}{{\rm\,dist}}
\ncom{\sspan}{{\rm\,span}}
\ncom{\re}{{\rm Re\,}}
\ncom{\im}{{\rm Im\,}}
\ncom{\sgn}{{\rm sgn\,}}
\ncom{\hone}{\mbox{\hspace{1em}}}
\ncom{\htwo}{\mbox{\hspace{2em}}}
\ncom{\hthree}{\mbox{\hspace{3em}}}
\ncom{\hfour}{\mbox{\hspace{4em}}}
\ncom{\vone}{\vskip 2ex}
\ncom{\vtwo}{\vskip 4ex}
\ncom{\vonee}{\vskip 1.5ex}
\ncom{\vthree}{\vskip 6ex}
\ncom{\vfour}{\vspace*{8ex}}
\ncom{\norm}{\|\;\;\|}
\ncom{\integ}[4]{\int_{#1}^{#2}\,{#3}\,d{#4}}
\ncom{\vspan}[1]{{{\rm\,span}\{ #1 \}}}
\ncom{\dm}[1]{ {\displaystyle{#1} } }
\ncom{\ri}[1]{{#1} \index{#1}}
\ncom{\vecs}{\mathbf{s}}
\ncom{\rvec}{\mathbf{r}}
\ncom{\rbvec}{\mathbf{\bar{r}}}
\ncom{\rmat}{\mathbf{R}}
\ncom{\rbmat}{\mathbf{\bar{R}}}
\ncom{\mmat}{\mathbf{M}}
\ncom{\mbmat}{\mathbf{\bar{M}}}
\ncom{\nmat}{\mathbf{N}}
\ncom{\zmat}{\mathbf{Z}}

\usepackage[T1]{fontenc}
\usepackage{lmodern}
\usepackage{bm}
\renewcommand{\vec}[1]{\bm{#1}}

\ncom{\vecone}{\vec{{1}}}
\ncom{\identity}{\vec{I}}
\ncom{\zero}{\vec{{0}}}
\ncom{\matone}{\vec{J}}
\ncom{\hatmat}{\vec{H}}
\ncom{\eps}{\epsilon}
\ncom{\cmat}{\vec{C}}
\ncom{\cbmat}{\vec{\bar{C}}}

\usepackage{titlesec}
 \usepackage{etoolbox, chngcntr}
\AtBeginEnvironment{appendices}{%
 \titleformat{\section}{\bfseries\Large}{\appendixname~\thesection:}{0.5em}{}%
 \titleformat{\subsection}{\bfseries\large}{\thesubsection}{0.5em}{}%
\counterwithin{equation}{section}
}

\ncom{\mybib}{\bibliography{references}\bibliographystyle{apalike}}

\newtheorem{remarks}{\bf Remark}[section]

\newtheorem{lemma}{\bf Lemma}[section]

\newtheorem{thm}{\bf Theorem}[section]

\numberwithin{equation}{section}

\makeatletter
\def\namedlabel#1#2{\begingroup
    #2%
    \def\@currentlabel{#2}%
    \phantomsection\label{#1}\endgroup
}
\makeatother

\makeatletter
\def\@makefnmark{%
  \leavevmode
  \raise.9ex\hbox{\fontsize\sf@size\z@\normalfont\tiny\@thefnmark}}
\makeatother

\makeatother

\providecommand{\keywords}[1]
{
  \small	
  \textbf{\textit{Keywords---}} #1
}

\usepackage{natbib}

\setcitestyle{aysep={}} 

\usepackage{placeins}

\usepackage{booktabs}
\usepackage{mathrsfs}

\begin{document}
\renewcommand{\subsectionautorefname}{subsection}

\pagenumbering{arabic}

\title{\bf Universally Optimal Multivariate Crossover Designs}
 \author{Shubham Niphadkar\thanks{ Department of Mathematics, Indian Institute of Technology Bombay, Mumbai 400 076, India}\hspace{.2cm} and Siuli Mukhopadhyay\footnotemark[1] \textsuperscript{,} \thanks{Corresponding author; Email: siuli@math.iitb.ac.in}\hspace{.2cm} }
\date{}
  \maketitle

\begin{abstract}
In this article, universally optimal multivariate crossover designs are studied. The multiple response crossover design is motivated by a $3 \times 3$ crossover setup, where the effect of $3$ doses of an oral drug are studied on gene expressions related to mucosal inflammation. Subjects are assigned to three treatment sequences and response measurements on $5$ different gene expressions are taken from each subject in each of the $3$ time periods. To model multiple or $g$ responses, where $g>1$, in a crossover setup, a multivariate fixed effect  model with both direct and carryover treatment effects is considered. It is assumed that there are non zero within response correlations, while between response correlations are taken to be zero. The information matrix corresponding to the direct effects is obtained and some results are studied. The information matrix in the multivariate case is shown to differ from the univariate case, particularly in the completely symmetric property. For the $g>1$ case, with $t$ treatments and $p$ periods, for $p=t \geq 3$, the  design represented by a Type $\rm{I}$ orthogonal array of strength $2$ is proved to be universally optimal over the class of binary designs, for the direct treatment effects.
\end{abstract} \hspace{10pt}

\keywords{Binary designs, Completely symmetric, Correlated response, Orthogonal arrays}

\section{Introduction}\label{sec1}
Often in clinical studies, we come across crossover trials that measure multiple responses from each subject in each period. For example, consider the blood sugar levels recorded at multiple time points in each period \citep{Putt1999AStudies} or multiple gene expression profiles of subjects measured in each period \citep{Leaker2017TheComplement, pareek2023likelihood}. Though these researchers discuss methods for estimating and analyzing the multiple responses measured in each period taking into account the various types of correlation structures which can exist between and within responses, they do not discuss any optimal design results for such multivariate crossover trials. In this article, our aim is to find universally optimal design for such multiple response crossover trials. The design represented by an orthogonal array of Type $\rm{I}$ and strength $2$ is shown to be universally optimal for the $p=t\geq 3$ case, where $p$ and $t$ denote the number of periods and treatments, respectively.\par

Crossover designs were initially developed to be used in agricultural sciences \citep{Cochran1939Long-termExperiments}. Later, these designs were also used in various other fields, such as pharmaceutical and clinical trials, and bioequivalence and biological studies \citep{Singh2016BayesianModels}. Optimal crossover designs for univariate responses in a fixed effect model setting have been studied by numerous researchers, namely \cite{Hedayat1975RepeatedI, Hedayat1978RepeatedII}, \cite{Cheng1980BalancedDesigns}, \cite{kunert1983optimal1, Kunert1984OptimalityDesigns}, \cite{Stufken1991SomeDesigns}, \cite{Kushner1997OptimalityDesigns, Kushner1998OptimalObservations}, \cite{Kunert2000OptimalityErrors}, \cite{Hedayat2003UniversalDesigns, Hedayat2004UniversalDesigns} and \cite{Singh2021EfficientSettings}. Whereas, crossover designs with random subject effects were explored by \cite{Laska1983OptimalEffects}, \cite{Laska1985AModels}, \cite{Carriere1993OptimalTreatments}, \cite{Carriere2000CrossoverTrials} and \cite{Hedayat2006OptimalRandom}, to name a few. Recently, there has been some interest in finding optimal crossover designs for generalized linear models, see \cite{Singh2016BayesianModels}, \cite{Jankar2020OptimalModels}, \cite{Mukhopadhyay2021LocallyDesigns} and \cite{Singh2021MinmaxTrials}. For a detailed review of crossover designs, we would like to refer to the paper by \cite{Bose2013DevelopmentsDesigns} and books by \cite{Senn2002Cross-overResearch}, \cite{Bose2009OptimalDesigns} and \cite{Kenward2014CrossoverTrials}. However, all these works concentrate only on single responses measured in each period. Till date, there has been no work on the development of universally optimal crossover designs for the multiple response case. 

In this article, to model data from a multiple response crossover trial, we propose a multivariate fixed effect model with both direct and carryover effects of treatments. The underlying assumption made is that no correlation exists between distinct responses. However, we allow for within response correlation, accounting for homoscedastic error variances. Information matrix for the direct effects is investigated and is shown to differ from the univariate case by its lack of complete symmetricity property thus violating the sufficient condition by \cite{Kiefer1975ConstructionIi}. To find optimal designs we instead resort to the more general technique proposed by \cite{yen1986conditions} for determining universally optimal multiple response crossover design.

\section{Motivating Example}\label{example data}
For motivating multiple response crossover trials, we use a gene expression dataset more recently considered by \cite{Leaker2017TheComplement}. In this study, a randomized, double-blind crossover experiment involving $3$ periods and $3$ treatments was considered. This crossover trial was placebo controlled and the goal was to study the effects of two single doses of oral drug, prednisone ($10$ mg and $25$ mg), with placebo on biomarkers of mucosal inflammation and transcriptomics after a nasal allergen challenge. The subjects enrolled in the study were assigned to one of the $3$ treatment sequences; $ABC$, $CAB$ and $BCA$. Treatment $A$ was the $10$ mg dose of the drug, while $B$ and $C$ were the placebo and $25$ mg dose of the drug, respectively. For our purpose, we considered $5$ gene profiles recorded in the nasal allergen challenge as the multiple outcomes.

Various tests were performed to check for correlation between and within the $5$ responses. Tables~\ref{table3}, \ref{table4a} and \ref{table5} show the test results. Table~\ref{table3} explores correlations between different responses in the same period, while Table~\ref{table4a} do this for different periods. Tables~\ref{table3} and \ref{table4a} show that there is no significant correlation (at a level of significance $0.01$) between distinct responses measured in the same and different time periods. In Table~\ref{table5}, we investigate if there is a presence of within response correlation, i.e., if the observations measured on the same genes are correlated. From the results given in Table~\ref{table5}, we note that some sample correlation coefficients have significantly low p-values (lower than $0.01$) thus implying within gene correlations.

\begin{table}[h]
\caption{Results for testing significance of within-period correlation coefficient between genes}\label{table3}
\begin{tabular*}{\textwidth}{@{\extracolsep\fill}lcccccc}
\toprule%
& \multicolumn{2}{@{}c@{}}{Period 1} & \multicolumn{2}{@{}c@{}}{Period 2} & \multicolumn{2}{@{}c@{}}{Period 3}\\\cmidrule{2-3}\cmidrule{4-5}\cmidrule{6-7}%
Pair of Variables & $r$ & p-value & $r$ & p-value & $r$ & p-value \\
\midrule
$\left( Gene_1, Gene_2 \right)$ &   $0.5249$ &  $0.0368$ & $0.2405 $ &  $0.3697$ & $0.5819$ &  $0.0181$\\
 $\left( Gene_1, Gene_3 \right)$& $-0.5014$ &  $0.0479$ & $-0.1112 $ &  $0.6819$ & $-0.1764$ &  $0.5135$\\
 $\left( Gene_1, Gene_4 \right)$& $-0.5898 $ &  $0.0162$ & $-0.5665 $ &  $0.0221$ & $0.0747 $ &  $0.7834$\\
 $\left( Gene_1, Gene_{5} \right)$& $0.0283$ &  $0.9172$ & $0.0184 $ &  $0.9461$ & $-0.2763 $ &  $0.3003$\\
 $\left( Gene_2, Gene_3 \right)$& $-0.5798$ &  $0.0186$ & $-0.5100$ &  $0.0436$ & $-0.3393$ &  $0.1985$\\
 $\left( Gene_2, Gene_4 \right)$& $-0.1938$ &  $0.472$ & $-0.0055$ &  $0.9838$ & $0.2499$ &  $0.3506$\\
 $\left( Gene_2, Gene_{5} \right)$& $-0.2325$ &  $0.3862$ & $0.0784 $ &  $0.773$ & $-0.2094$ &  $0.4364$\\
 $\left( Gene_3, Gene_4 \right)$& $0.1361$ &  $0.6152$ & $0.5072$ &  $0.0449$ & $0.3269$ &  $0.2166$\\
 $\left( Gene_3, Gene_{5} \right)$& $-0.0201$ &  $0.9391$ & $-0.1387$ &  $0.6085$ & $-0.2275$ &  $0.3969$\\
 $\left( Gene_4, Gene_{5} \right)$& $-0.2338$ &  $0.3834$ & $-0.1856$ &  $0.4913$ & $0.1346$ &  $0.6191$\\
\bottomrule
\end{tabular*}
\end{table}

\begin{sidewaystable}[p]
\caption{Results for testing significance of between-period correlation coefficient between genes}\label{table4a}
\begin{tabular*}{\textwidth}{@{\extracolsep\fill}lcccccccccccc}
\cmidrule{1-13}%
& \multicolumn{2}{@{}c@{}}{{\small{(Period 1, Period 2)}}} & \multicolumn{2}{@{}c@{}}{{\small{(Period 1, Period 3)}}} & \multicolumn{2}{@{}c@{}}{{\small{(Period 2, Period 1)}}} & \multicolumn{2}{@{}c@{}}{{\small{(Period 2, Period 3)}}} & \multicolumn{2}{@{}c@{}}{{\small{(Period 3, Period 1)}}} & \multicolumn{2}{@{}c@{}}{{\small{(Period 3, Period 2)}}}\\\cmidrule{2-3}\cmidrule{4-5}\cmidrule{6-7}\cmidrule{8-9}\cmidrule{10-11}\cmidrule{12-13}%
{\small \small{Pair of Variables}} & {\small \small{$r$}} & {\small \small{p-value}} & {\small \small{$r$}} & {\small \small{p-value}} & {\small \small{$r$}} & {\small \small{p-value}} & {\small \small{$r$}} & {\small \small{p-value}} & {\small \small{$r$}} & {\small \small{p-value}} & {\small \small{$r$}} & {\small \small{p-value}}\\
\cmidrule{1-13}%
{\small{$\left( Gene_1, Gene_2 \right)$}} &   {\small{$0.6019$}} &  {\small{$0.0136$}} & {\small{$0.4356 $}} &  {\small{$0.0917$}} & {\small{$0.3616$}}  & {\small{$0.1687$}} & {\small{$0.3186$}} &  {\small{$0.2292$}} & {\small{$0.2368$}} &  {\small{$0.3772$}} & {\small{$0.2291$}} &  {\small{$0.3935$}}\\
 {\small{$\left( Gene_1, Gene_3 \right)$}} & {\small{$-0.3796$}} &  {\small{$0.147$}} & {\small{$-0.3110 $}} &  {\small{$0.241$}} & {\small{$-0.3812$}} &  {\small{$0.1451$}} & {\small{$0.0039$}} &  {\small{$0.9887$}} & {\small{$-0.4121 $}} &  {\small{$0.1127$}} & {\small{$-0.0867$}} &  {\small{$0.7497$}}\\
 {\small{$\left( Gene_1, Gene_4 \right)$}} & {\small{$-0.3925 $}} &  {\small{$0.1327$}} & {\small{$0.1514$}} &  {\small{$0.5758$}} & {\small{$-0.3572$}} &  {\small{$0.1744$}} & {\small{$0.5284$}} &  {\small{$0.0354$}} & {\small{$-0.2977$}} &  {\small{$0.2628$}} & {\small{$-0.3366 $}} & {\small{$0.2024$}}\\
{\small{$\left( Gene_1, Gene_{5} \right)$}} & {\small{$0.1760 $}} & {\small{$0.5145$}} & {\small{$0.1070 $}} &  {\small{$0.6932$}} & {\small{$-0.2620 $}} &  {\small{$0.327$}} & {\small{$0.2173 $}} &  {\small{$0.4189$}} & {\small{$-0.1770 $}} &  {\small{$0.512$}} & {\small{$0.2731 $}} &  {\small{$0.3061$}}\\
{\small{$\left( Gene_2, Gene_3 \right)$}} & {\small{$-0.2346$}} &  {\small{$0.3819$}} & {\small{$-0.1038$}} &  {\small{$0.702$}} & {\small{$-0.4676$}} &  {\small{$0.0678$}} & {\small{$-0.2482 $}} & {\small{$0.3539$}} & {\small{$-0.4988 $}} &  {\small{$0.0492$}} & {\small{$-0.2914$}} &  {\small{$0.2736$}}\\
{\small{$\left( Gene_2, Gene_4 \right)$}} & {\small{$-0.2799 $}} &  {\small{$0.2938$}} & {\small{$0.2380 $}} &  {\small{$0.3748$}} & {\small{$-0.6061$}} &  {\small{$0.0128$}} & {\small{$0.2259$}} &  {\small{$0.4002$}} & {\small{$-0.1550$}} &  {\small{$0.5664$}} & {\small{$-0.3463$}} &  {\small{$0.1889$}}\\
{\small{$\left( Gene_2, Gene_{5} \right)$}} & {\small{$0.0004 $}} &  {\small{$0.999$}} & {\small{$0.1501$}} &  {\small{$0.5789$}} & {\small{$0.2867$}} &  {\small{$0.2817$}} & {\small{$0.1587$}} &  {\small{$0.5573$}} & {\small{$-0.2851$}} &  {\small{$0.2844$}} & {\small{$0.0508$}} &  {\small{$0.8517$}}\\
 {\small{$\left( Gene_3, Gene_4 \right)$}} & {\small{$-0.0881 $}} &  {\small{$0.7456$}} & {\small{$0.0819 $}} &  {\small{$0.763$}} & {\small{$0.3954$}} &  {\small{$0.1296$}} & {\small{$-0.5053$}} &  {\small{$0.0459$}} & {\small{$-0.0110 $}} &  {\small{$0.9678$}} & {\small{$0.0025 $}} &  {\small{$0.9928$}}\\
{\small{$\left( Gene_3, Gene_{5} \right)$}} & {\small{$0.1892$}} &  {\small{$0.4829$}} & {\small{$-0.0959 $}} &  {\small{$0.7239$}} & {\small{$-0.2970 $}} &  {\small{$0.264$}} & {\small{$-0.0011$}} &  {\small{$0.9967$}} & {\small{$-0.0234 $}} &  {\small{$0.9315$}} & {\small{$-0.1645$}} &  {\small{$0.5426$}}\\
{\small{$\left( Gene_4, Gene_{5} \right)$}} & {\small{$-0.2483$}} &  {\small{$0.3539$}} & {\small{$-0.1959$}} &  {\small{$0.4672$}} & {\small{$0.1379 $}} &  {\small{$0.6105$}} & {\small{$0.2486$}} &  {\small{$0.3532$}} & {\small{$0.0785 $}} &  {\small{$0.7726$}} & {\small{$-0.1141 $}} &  {\small{$0.674$}}\\
\cmidrule{1-13}%
\end{tabular*}
\end{sidewaystable}

\begin{table}[h]
\caption{Results for testing significance of between-period correlation coefficient within genes}\label{table5}
\begin{tabular*}{\textwidth}{@{\extracolsep\fill}lcccccc}
\toprule%
& \multicolumn{2}{@{}c@{}}{(Period 1, Period 2)} & \multicolumn{2}{@{}c@{}}{(Period 1, Period 3)} & \multicolumn{2}{@{}c@{}}{(Period 2, Period 3)}\\\cmidrule{2-3}\cmidrule{4-5}\cmidrule{6-7}%
Variable & $r$ & p-value & $r$ & p-value & $r$ & p-value \\
\midrule
 $ Gene_1$ &  $0.6145$  &  $0.0113$ & $0.6069$ &  $0.0127$ & $0.4889$ & $0.0547$\\
 $ Gene_2$ & $0.3994$ &   $0.1253$ & $0.5722$ &  $0.0206$ & $0.5323$ &  $0.0338$\\
 $ Gene_3$ &  $-0.1713$ &   $0.5258$ & $0.2573$ & $0.336$ & $0.0626$ & $0.8178$\\
 $ Gene_4$ &  $0.0489$ &  $0.8574$ & $-0.2812$ & $0.2914$ & $-0.6729$ &  $0.0043$\\
 $ Gene_{5}$ &  $-0.1912$ & $0.478$ & $0.0054$ &  $0.9841$ & $-0.3650$ &  $0.1645$\\
\bottomrule
\end{tabular*}
\end{table}

We use the above motivating example of the genetic study and the corresponding correlation results to frame the multivariate crossover statistical model in the next section.

\section{Proposed Statistical Model}\label{proposed-model}

We consider crossover designs with $t$ treatments, $n$ subjects and $p$ periods, where $t,p \geq 2$. Let $\varOmega_{t,n,p}$ be the class of all such designs. We consider that for any design in $\varOmega_{t,n,p}$ and $g \geq 1$, there are $g$ response variables on which observations are recorded corresponding to every subject in each period. For $1 \leq k \leq g$, corresponding to the $k^{th}$ response, let $Y_{dijk}$ represent the random variable corresponding to the observation from the $i^{th}$ period and the $j^{th}$ subject, where $1 \leq i \leq p$ and $1 \leq j \leq n$. We consider that $Y_{dijk}$ satisfies the following model:
\begin{align}
Y_{dijk} = \mu_k + \alpha_{i,k} + \beta_{j,k} + \tau_{d \left( i,j \right),k} + \rho_{d \left( i-1,j \right),k}  + \eps_{ijk},
\label{model1}
\end{align}
where corresponding to the $k^{th}$ response variable, $\eps_{ijk}$ is the error term with mean $0$ and variance $\sigma^2$, $\mu_k$ is the intercept, $\alpha_{i,k}$ is the $i^{th}$ period effect, $\beta_{j,k}$ is the $j^{th}$ subject effect and for $1 \leq s \leq t$, $\tau_{s,k}$ is the direct effect and $\rho_{s,k}$ is the first order carryover effect due to the $s^{th}$ treatment. Here $d(i,j)$ represents the treatment allocated to the $i^{th}$ period and the $j^{th}$ subject. We assume that the model \eqref{model1} is a fixed effect model. In the above model, all the effects are taken to vary along with the response variable and no carryover effect is assumed in  the first period.

Expressing the above model equation in matrix notations, we obtain
\begin{equation}
\vec{Y}_{dk} = \vec{X}_{dk} \vec{\theta}_k + \vec{\eps}_k,
\label{model11}
\end{equation}
where $\vec{Y}_{dk} = \left( Y_{d11k}, \cdots, Y_{dp1k}, Y_{d12k}, \cdots, Y_{dp2k}, \cdots, Y_{d1nk}, \cdots, Y_{dpnk} \right)^{'}$, $\vec{X}_{dk}=
\begin{bmatrix}
\vecone_{np} & \vec{P} & \vec{U} & \vec{T}_d & \vec{F}_d
\end{bmatrix}$, $\vec{\eps}_k = \left( \eps_{11k}, \cdots, \eps_{p1k}, \eps_{12k}, \cdots, \eps_{p2k}, \cdots, \eps_{1nk}, \cdots, \eps_{pnk} \right)^{'}$, and $\vec{\theta}_k = \left( \mu_k, \vec{\alpha}_k^{'}, \vec{\beta}_k^{'}, \vec{\tau}_k^{'}, \vec{\rho}_k^{'} \right)^{'}$ is the parameter vector of length $(1+p+n+2t)$, with components: $\mu_k$, $\vec{\alpha}_k = \left( \alpha_{1,k}, \cdots, \alpha_{p,k} \right)^{'}$, $\vec{\beta}_k = \left( \beta_{1,k}, \cdots, \beta_{n,k} \right)^{'}$, $\vec{\tau}_k = \left( \tau_{1,k}, \cdots, \tau_{t,k} \right)^{'}$ and $\vec{\rho}_k = \left( \rho_{1,k}, \cdots, \rho_{t,k} \right)^{'}$. Here, the period and subject effects are accounted by, $\vec{P}=\vecone_n \otimes \identity_p$, $\vec{U}=\identity_n \otimes \vecone_p$, respectively, while $\vec{T}_d$ is the design matrix corresponding to treatment effects, and $\vec{F}_d$ is the design matrix for carryover effects. Note that $\vec{F}_d = \left(\identity_n \otimes \vec{\psi}\right)  \vec{T}_d$, where $\vec{\psi} =  
\begin{bmatrix}
\zero^{'}_{p-1 \times 1} & 0\\
\identity_{p-1} & \zero_{p-1 \times 1}
\end{bmatrix}$. We partition $\vec{X}_{dk}$ as 
\begin{equation}
\vec{X}_{dk} =
\begin{bmatrix}
\vecone_{np} & \vec{X}_1 & \vec{X}_2
\end{bmatrix},
\label{Xd}
\end{equation}
where $\vec{X}_1 = 
\begin{bmatrix} 
\vec{P} & \vec{U}
\end{bmatrix}$ and $\vec{X}_2 = 
\begin{bmatrix}  
\vec{T}_d & \vec{F}_d 
\end{bmatrix}$.
Combining the model equations for all $g$ responses, we get
\begin{align}
\begin{bmatrix}
 \vec{Y}^{'}_{d1} &
\cdots &
\vec{Y}^{'}_{dg}
\end{bmatrix}^{'}
&= \left( \identity_{g} \otimes \vec{X}_{d1} \right) 
\begin{bmatrix}
\vec{\theta}^{'}_{1} &
\cdots &
\vec{\theta}^{'}_g
\end{bmatrix}^{'}
+
\begin{bmatrix}
 \vec{\eps}^{'}_{1} &
\cdots &
\vec{\eps}^{'}_g
\end{bmatrix}^{'},
\label{unio}
\end{align}
where $\mathbb{E} \left( \vec{\eps}_k \right) = \zero_{np \times 1}$. It is assumed that (i) the observations from different subjects  are uncorrelated, and (ii) no correlations exists between distinct response variables implying, $\mathbb{\text{Cov}} \left( \vec{\eps}_k, \vec{\eps}_{k'} \right) = \zero_{np \times np}$, for $1 \leq k \neq k' \leq g$. Under these assumptions, the dispersion matrix of $\vec{\eps}_k$ is taken to be
\begin{align}
\mathbb{D} \left( \vec{\eps}_k \right) = \sigma^2 \vec{\varSigma} = \sigma^2 \left( \identity_{n} \otimes \vec{V} \right),
\end{align}
where $\sigma^2 >0$ is an unknown constant, while $\vec{V}$ is assumed to be a known positive definite and symmetric $p \times p$ matrix. 

\section{Information Matrix of the Direct Treatment Effects}
\label{information matrices}
In this section, we determine the information matrix under model \eqref{unio} for the direct effects, for $g \geq 1$. We use some results for the $g=1$ case from \cite{Bose2009OptimalDesigns}.

For any design $d \in \varOmega_{t,n,p}$, let $\vec{C}_{d}$ represent the information matrix for the direct effects. Theorem~{\upshape\ref{thm4b}} is the main result for the information matrix in the $g \geq 1$ case. 
\begin{thm}
The information matrix for the direct effects can be expressed as
\begin{align}
\vec{C}_{d}&= \identity_g \otimes \left[ \vec{C}_{{d}11} - \vec{C}_{{d}12}\vec{C}^{-}_{{d}22} \vec{C}_{{d}21} \right], \label{p9}
\end{align}
where $\vec{C}_{{d}11} = \vec{T}^{'}_d \vec{A}^{*} \vec{T}_d$, $\vec{C}_{{d}12} = \vec{C}_{{d}21}^{'} = \vec{T}^{'}_d \vec{A}^{*} \vec{F}_d$, $\vec{C}_{{d}22} = \vec{F}^{'}_d \vec{A}^{*} \vec{F}_d$, $\vec{A}^{*} = \hatmat_n \otimes \vec{V}^{*}$ and $\vec{V}^{*} = \vec{V}^{-1} - \left( \vecone_p^{'} \vec{V}^{-1} \vecone_p \right)^{-1} \vec{V}^{-1} \vecone_{p} \vecone_{p}^{'} \vec{V}^{-1}$. Here $\hatmat_n = \identity_n - \frac{1}{n} \vecone_n \vecone_n^{'}$ and $\vec{C}^{-}_{d22}$ is a generalized inverse of of $\vec{C}_{d22}$.
\label{thm4b}
\end{thm}
\begin{proof}
\renewcommand{\subsectionautorefname}{Supplementary Material}
For details of the proof, refer to Appendix \ref{secB}.
\end{proof}

\begin{remarks}
The information matrix $\vec{C}_d$ is symmetric, non-negative definite (n.n.d.) matrix having zero row sums and column sums, and is invariant with respect to the choice of generalized inverses involved.
\label{re2a}
\end{remarks}
\begin{proof}
\renewcommand{\subsectionautorefname}{Supplementary Material}
The proof is given in Appendix \ref{secB}.
\end{proof}

\section{Universal Optimality}
For the single response case, universal optimality criterion have been discussed by various authors including \cite{Kiefer1975ConstructionIi} and \cite{Bose2009OptimalDesigns} (see Chapter 1, pp. 18--22). Here, we extend the definition of universally optimal designs by \cite{Kiefer1975ConstructionIi} to the $g>1$ setup as follows:\\
Let for $g > 1$, $\mathcal{B}_{gt}$ be a class of $gt \times gt$ symmetric, non-negative definite (n.n.d.) matrices having zero row sums. Suppose $\varPhi$ is the class of functions $\phi: \mathcal{B}_{gt} \rightarrow \left( -\infty , \infty \right]$, such that
\begin{enumerate}
\item \label{aa} $\phi$ is matrix convex;
\item \label{bb} For any matrix $\vec{A} \in \mathcal{B}_{gt}$, $\phi \left( x \vec{A}  \right)$ is non-increasing in the scalar $x\geq 0$;
\item \label{cc}$\phi$ is invariant under each simultaneous permutation of rows and columns of $\vec{A} \in \mathcal{B}_{gt}$.
\end{enumerate}
Then a design $d^* \in \mathcal{D}$, where $\mathcal{D}$ is a subclass of designs, is a universally optimal design for the parameters of interest over $\mathcal{D}$, if $d^*$ minimizes $\phi \left( \vec{A}_d  \right)$ over $\mathcal{D}$, where $\vec{A}_d$ is the information matrix for the parameters of interest.

Note in the multivariate setup, we have $gt$ direct treatment effect parameters and the information matrix $\vec{A}_d$ is of order $gt \times gt$. In the single response case, the sufficient conditions of \cite{Kiefer1975ConstructionIi} (a) complete symmetricity of the information matrix corresponding to $d^*$, and (b) $d^*$ maximizing the trace of the information matrix over $\mathcal{D}$, where $\mathcal{D}$ is a subclass of designs, leads to the universal optimality of $d^*$ over $\mathcal{D}$. However in the $g>1$ setup, due to the lack of complete symmetricity of the information matrix for the direct effects (see Lemma~{\upshape\ref{ree2c}} in the Appendix \ref{secA}), we are unable to use a similar sufficient condition. For our multivariate setup, instead we resort to the use of a sufficient condition on the lines of \cite{yen1986conditions} as stated in the next Lemma.
\begin{lemma}
Let a design $d^* \in \mathcal{D}$ be such that for $d \in \mathcal{D}$, the corresponding information matrix $\vec{A}_d \in \mathcal{B}_{gt}$; for any $d \in \mathcal{D}$ and $\vec{A}_{d} \neq \zero_{gt \times gt}$, there exists scalars $b_{d1}, \cdots, b_{d(gt)} \geq 0$ satisfying $\vec{A}_{d^*} = \sum_{\kappa=1}^{(gt)!} b_{d\kappa} \vec{P}_{\kappa} \vec{A}_{d} \vec{P}^{'}_{\kappa}$, where $g>1$ and $\vec{P}_1$, $\cdots$, $\vec{P}_{(gt)!}$ are all possible distinct $gt \times gt$ permutation matrices. Then the design $d^*$ is universally optimal for the parameters of interest over $\mathcal{D}$, if $d^*$ maximizes $\text{\rmfamily\upshape tr} \left( \vec{A}_{d} \right)$ over $d \in \mathcal{D}$.
\label{lemma1}
\end{lemma}
\begin{proof}
The proof is on similar lines as in \cite{yen1986conditions}. Please see Appendix \ref{secA} for an outline of the proof.
\end{proof}

\subsection{Universal Optimality of Orthogonal Array Designs}
In this subsection, we consider designs with $p=t$. Our interest is to check if the universal optimality of a design represented by an orthogonal array of Type $\rm{I}$ and strength $2$, $OA_{I} \left( n=\lambda t \left(t-1 \right), p=t, t, 2 \right)$, where $t \geq 3$ and $\lambda$ is a positive integer, holds for the $g>1$ case. Note that in the univariate case with correlated error terms, a design given as $OA_{I} \left( n=\lambda t \left(t-1 \right), p, t, 2 \right)$, where $3 \leq p \leq t$ and $\lambda$ is a positive integer, is shown to be universally optimal for the direct effects over the class of binary designs by \cite{Kunert2000OptimalityErrors}. 
\begin{thm}
Let $d^* \in \mathcal{D}$ be a design given by $OA_{I} \left( n=\lambda t \left(t-1 \right), p=t, t, 2 \right)$, where $\mathcal{D}$ is a class of binary designs with $p=t$, $\lambda$ is a positive integer and $t \geq 3$. Then $d^*$ is a universally optimal design for the direct effects over $\mathcal{D}$ in the $g > 1$ case.
\label{thm1-c3}
\end{thm}
\begin{proof}
\label{proof5}
Here $d^* \in \mathcal{D}$ is a $OA_{I} \left( n=\lambda t \left(t-1 \right), p=t, t, 2 \right)$, where $\mathcal{D}$ is a class of binary designs with $p=t$, $\lambda$ is a positive integer and $t \geq 3$. From Remark~{\upshape\ref{re2a}}, it can be clearly seen that for $d \in \mathcal{D}$, $\vec{C}_{d} \in \mathcal{B}_{gt}$. Also, $d^*$ is a universally optimal design over $\mathcal{D}$ under the $g=1$ case \citep{Kunert2000OptimalityErrors}. Thus we know that for $g=1$, $d^*$ maximizes $\text{\rmfamily\upshape tr} \left( \vec{C}_{{d}}  \right) = \text{\rmfamily\upshape tr} \left( \vec{C}_{{d}11} - \vec{C}_{{d}12} \vec{C}^{-}_{{d}22} \vec{C}_{{d}21} \right)$ over $\mathcal{D}$. From Theorem~{\upshape\ref{thm4b}}, for $g>1$ and $d \in \mathcal{D}$, the information matrix for the direct effects can be written as
\begin{align}
\vec{C}_{d}&= \identity_g \otimes \left[ \vec{C}_{{d}11} - \vec{C}_{{d}12} \vec{C}^{-}_{{d}22} \vec{C}_{{d}21} \right].
\label{eq5}
\end{align}
Hence for $g>1$ also, $d^*$ maximizes $\text{\rmfamily\upshape tr} \left( \vec{C}_{{d}}  \right)$ over $\mathcal{D}$.

From \cite{Martin1998Variance-balancedObservations} and \cite{Kunert2000OptimalityErrors}, we get that for $g=1$, the information matrix $\vec{C}_{d^*}$ is a completely symmetric matrix. So using Remark 1 from \cite{yen1986conditions}, we get for $g=1$, any $d \in \mathcal{D}$ and $\vec{C}_{{d}11} - \vec{C}_{{d}12} \vec{C}^{-}_{{d}22} \vec{C}_{{d}21} \neq \zero_{t \times t}$,
\begin{equation}
\begin{split}
\vec{C}_{d^*} &= \vec{C}_{{d^*}11} - \vec{C}_{{d^*}12} \vec{C}^{-}_{{d^*}22} \vec{C}_{{d^*}21} = \sum_{\upsilon=1}^{t!}  c_{d\upsilon}\vec{Q}_{\upsilon} \left( \vec{C}_{{d}11} - \vec{C}_{{d}12} \vec{C}^{-}_{{d}22} \vec{C}_{{d}21} \right) \vec{Q}^{'}_{\upsilon},
\end{split}
\label{eq7}
\end{equation}
where $c_{d\upsilon} = \left[\text{\rmfamily\upshape tr} \left( \vec{C}_{d^*} \right) \right]/\left[ t! \times \text{\rmfamily\upshape tr} \left( \vec{C}_{{d}11} - \vec{C}_{{d}12} \vec{C}^{-}_{{d}22} \vec{C}_{{d}21} \right) \right] \geq 0$ and $\vec{Q}_{\upsilon}$'s are distinct $t \times t$ permutation matrices.

Using \eqref{eq5}, it is clear that for $g>1$ and $d \in \mathcal{D}$, $\vec{C}_{d} = \zero_{gt \times gt}$ if and only if $\vec{C}_{{d}11} - \vec{C}_{{d}12} \vec{C}^{-}_{{d}22} \vec{C}_{{d}21} = \zero_{t \times t}$.
For $g>1$, any $d \in \mathcal{D}$, $\vec{C}_{d} \neq \zero_{gt \times gt}$ and $\kappa=1, \cdots, (gt)!$, if we suppose
\begin{align}
b_{d\kappa} &= 
\begin{cases}
c_{d\upsilon}, \text{ if }\vec{P}_{\kappa} = \identity_g \otimes \vec{Q}_{\upsilon}, \text{ for some } \upsilon,\\
0, \text{ otherwise},
\end{cases}
\label{eq8}
\end{align}
then using \eqref{eq5} and \eqref{eq7}, we get that for $g>1$, any $d \in \mathcal{D}$ and $\vec{C}_{d} \neq \zero_{gt \times gt}$,
\begin{align}
\vec{C}_{d^*} &= \sum_{\kappa=1}^{(gt)!} b_{d\kappa} \vec{P}_{\kappa} \vec{C}_{d} \vec{P}^{'}_{\kappa}.
\label{eq9}
\end{align}
Thus from \eqref{eq9} and Lemma~{\upshape\ref{lemma1}}, we have proved $d^*$ is universally optimal for the direct effects over $\mathcal{D}$ in the $g>1$ case.
\end{proof}

\section{Illustration}
\renewcommand{\sectionautorefname}{Section}
In this section, we evaluate the efficiency of the $3 \times 3$ binary crossover design used in the genetic study by \cite{Leaker2017TheComplement}. As our response variates, we take the $5$ gene profiles. Each of the $3$ treatment sequences, $ABC$, $CAB$ and $BCA$, are applied on $6$ subjects, where treatment $A$ denotes the $10$ mg dose of the drug, $B$ denotes the placebo and $C$ denotes the $25$ mg dose of the drug. Let us denote the above considered design by $d_0$. To study the efficiency of $d_0$, we consider two different structures of matrix $\vec{V}$ as follows:
\renewcommand{\sectionautorefname}{Appendix}
\begin{enumerate}
\item \label{a1}
$\vec{V} = \left( 1 - r_{(1)}^2 \right)^{-1}
\begin{bmatrix}
1 & r_{(1)} & r_{(1)}^2\\
r_{(1)} & 1 & r_{(1)}\\
r_{(1)}^2 & r_{(1)} & 1
\end{bmatrix}$, where $-1 < r_{(1)} <1$;
\item \label{b1}
$\vec{V} =
\begin{bmatrix}
1 & r_{(1)} & 0\\
r_{(1)} & 1 & r_{(1)}\\
0 & r_{(1)} & 1
\end{bmatrix}$, where $-1/\sqrt{2} < r_{(1)} <1/\sqrt{2}$.
\end{enumerate}
\ref{a1} is a $AR(1)$ structure, and \ref{b1} is a tridiagonal structure with all diagonal and off-diagonal elements equal. From our choices of $\vec{V}$, we note it is a function of $r_{(1)}$ only, where $r_{(1)}$ is such that the matrix $\vec{V}$ is positive definite. Note that if $r_{(1)} = 0$, then both $AR(1)$ and tridiagonal structure reduces to a diagonal structure.

From Theorem~{\upshape\ref{thm1-c3}}, we know that a design $d^* \in \mathcal{D}$ represented by $OA_{I} \left( n=\lambda t \left(t-1 \right), p=t,\right.$\\$\left.t, 2 \right)$, where $\mathcal{D}$ is a class of binary designs with $p=t$, $\lambda$ is a positive integer and $t \geq 3$, is universally optimal for the direct effects over $\mathcal{D}$ for the $g>1$ setup. Suppose to get a measure of the efficiency of any binary design $d \in \mathcal{D}$ with $t$ and $n$ same as that of $d^*$ we use the ratio of the traces as follows:
\begin{align}
e &= \frac{\text{\rmfamily\upshape tr} \left( \vec{C}_d \right)}{\text{\rmfamily\upshape tr} \left( \vec{C}_{d^*} \right)}.
\end{align}
Values of $e$ close to $1$ shows that design, $d$, is a nearly efficient design. Figures~\ref{AR(1)} and \ref{Tridiagonal} plot $e$ under both the structures of $\vec{V}$. From Figures~\ref{AR(1)} and \ref{Tridiagonal}, we see that the maximum efficiency of design $d_0$ is $0.0278$ as compared to the universal optimal design $d^*$ for both structures of $\vec{V}$. Thus, $d_0$ is not an efficient/nearly efficient design.

\renewcommand{\thefigure}{\arabic{figure}}
\begin{figure}[H]
\centering
 \includegraphics[width=0.9\textwidth]{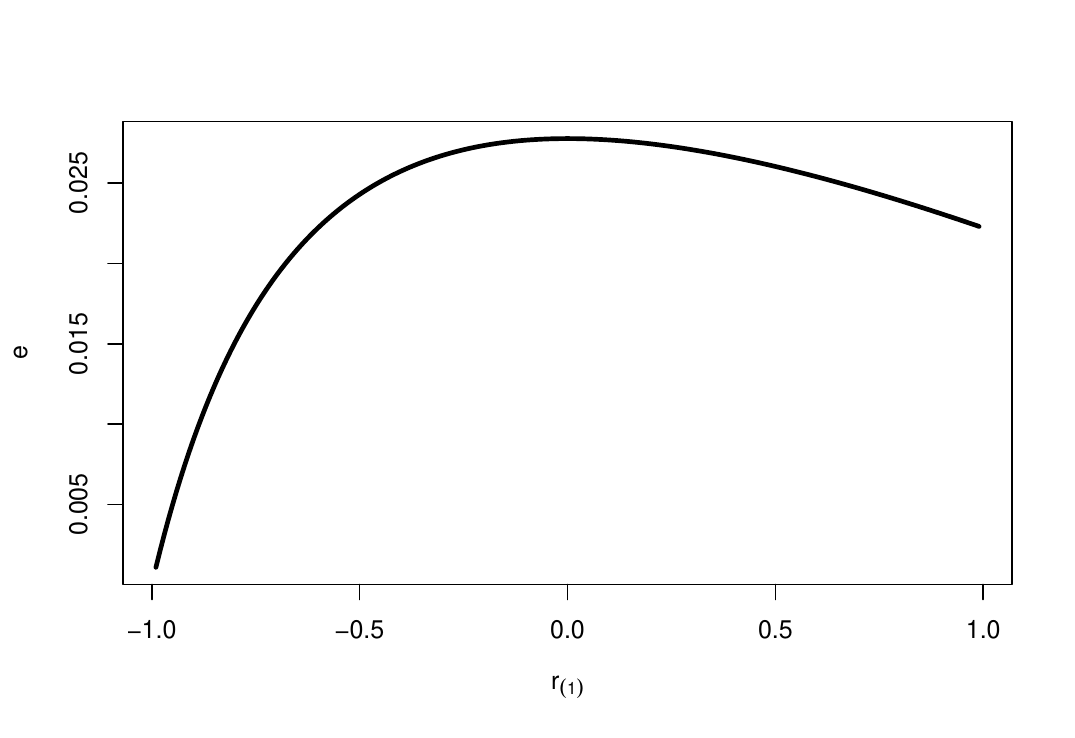}
 \caption{Efficiency of $d_0$ when $\vec{V}$ has a $AR(1)$ structure}
 \label{AR(1)}
\end{figure}
\begin{figure}[H]
  \centering
  \includegraphics[width=0.9\textwidth]{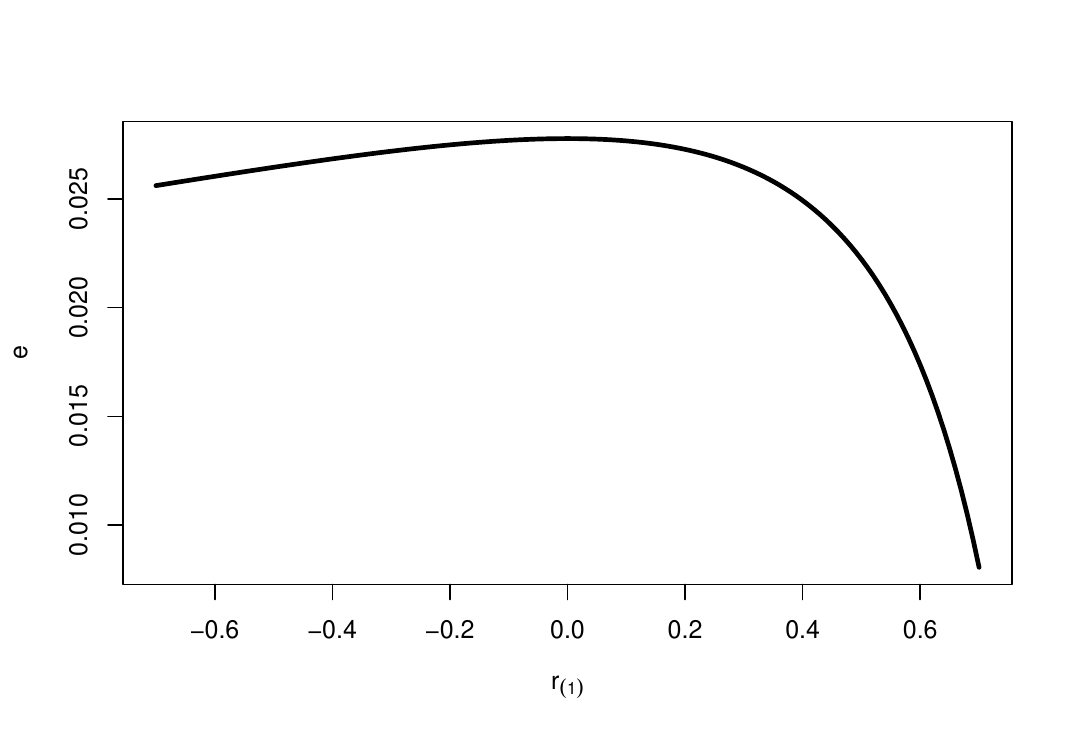}
  \caption{Efficiency of $d_0$ when $\vec{V}$ has a tridiagonal structure}
  \label{Tridiagonal}
\end{figure}

\section{Conclusion and Future Direction}
In this article, we investigated universal optimality for an orthogonal array of Type $\rm{I}$ and strength $2$ when $g$ responses are recorded in each period from each subject, where $g \geq 1$. Under the multivariate fixed effect model, the information matrix for the direct effects differed from the $g=1$ case, particularly in terms of the completely symmetric property. For the $g>1$ case and non-zero within response correlation, for $p=t \geq 3$, we showed a design given as an orthogonal array of type $\rm{I}$ and strength $2$ is universally optimal for the direct effects over a class of binary designs. By following similar techniques as employed in this article, we can also show that for uncorrelated and homoscedastic errors, if a balanced uniform design is universally optimal for the direct effects (carryover effects) over a subclass of designs for $g=1$, then the universal optimality also holds over the same subclass of designs for the $g >1$ case.

Though we did not consider that the between response correlation is measured in the same/different periods in this article, such correlation may exist in a multiple response crossover experiment. However, in those cases, the error covariances will be of complex nature and it may be tedious to determine theoretical optimal designs. As a future direction, we plan to investigate optimality results in such correlated crossover scenarios. In the future, we also plan to study the effect of heteroscedastic error terms on the results for universal optimality in the $g>1$ case.

\section*{Statements and Declarations}

\begin{itemize}
\item \textbf{Funding} The authors did not receive support from any organization for the submitted work.
\item \textbf{Competing interests} The authors declare no competing interests.
\item \textbf{Ethics approval and consent to participate} This research does not involve human or animal participants.
\item \textbf{Consent for publication}  Not applicable.
\item \textbf{Data availability} The data is available in the Gene Expression Omnibus database under accession number GSE67200 and is available at the following URL: \url{https://www.ncbi.nlm.nih.gov/geo/query/acc.cgi?acc=GSE67200}. 
\item \textbf{Materials availability} Not applicable.
\item \textbf{Code availability} \url{https://github.com/rsphd/Correlation-Code} contains the code (.R file) and the dataset on $5$ genes (.zip file consisting of 5 Excel files), for checking the significance of correlation between and within genes.
\item \textbf{Author contribution} All authors contributed to the study conception and design. Material preparation, analysis and investigation were performed by Shubham Niphadkar. The first draft of the manuscript was written by Shubham Niphadkar and all authors commented on previous versions of the manuscript. All authors read and approved the final manuscript.
\end{itemize}

\begin{appendices}
\numberwithin{equation}{section}
\section{Some Useful Proofs and Results}\label{secA}
\begin{proof}[Proof of Lemma~{\upshape\ref{lemma1}}]
\label{proof3}
Here $d \in \mathcal{D}$, where $\mathcal{D}$ is a subclass of designs such that $\vec{A}_d \in \mathcal{B}_{gt}$. For $g>1$, if $\vec{A}_d = \zero_{gt \times gt}$, then from condition~{\upshape\ref{bb}}, it is clear that $\phi \left( \vec{A}_{d^*} \right) \leq \phi \left( \vec{A}_{d} \right)$. Now we consider the case when for $g>1$, $\vec{A}_d \neq \zero_{gt \times gt}$. Note that for $d \in \mathcal{D}$, we have $\vec{A}_d \in \mathcal{B}_{gt}$. So $\vec{A}_d = \vec{L} \vec{L}^{'}$, where $\vec{L} = \vec{A}_d^{1/2}$. Hence we get that $\text{\rmfamily\upshape tr} \left( \vec{A}_d \right) >0$. Here we have $\vec{A}_{d^*} = \sum_{\kappa=1}^{(gt)!} b_{d\kappa} \vec{P}_{\kappa} \vec{A}_{d} \vec{P}^{'}_{\kappa}$. Since $\vec{A}_d \in \mathcal{B}_{gt}$, from condition~{\upshape\ref{cc}} we get that $\phi \left( \vec{A}_d \right) = \phi \left( \vec{P}_{\kappa} \vec{A}_{d} \vec{P}^{'}_{\kappa} \right)$, for $\kappa = 1, \cdots, (gt)!$. So, we get 
\begin{align}
\text{\rmfamily\upshape tr} \left( \vec{A}_{d^*} \right) = b_d \text{\rmfamily\upshape tr} \left(  \vec{A}_{d} \right),
\label{eqq2c}
\end{align}
where $b_d = b_{d1} + \cdots b_{d(gt)!}$. Also, here we know that $d^*$ maximizes $\text{\rmfamily\upshape tr} \left(  \vec{A}_{d} \right)$ over $\mathcal{D}$. Thus we get that $\text{\rmfamily\upshape tr} \left(  \vec{A}_{d^*} \right) \geq \text{\rmfamily\upshape tr} \left(  \vec{A}_{d} \right)$. Since $\text{\rmfamily\upshape tr} \left( \vec{A}_d \right) > 0$, from \eqref{eqq2c}, we get that $b_d \geq 1$. Thus using conditions~{\upshape\ref{aa}} and {\upshape\ref{bb}} along with $\vec{A}_{d^*} = \sum_{\kappa=1}^{(gt)!} b_{d\kappa} \vec{P}_{\kappa} \vec{A}_{d} \vec{P}^{'}_{\kappa}$, we get
\begin{align*}
\phi \left( \vec{A}_{d^*} \right) = \phi \left( b_d \sum_{\kappa=1}^{(gt)!} \left( b_{d\kappa}/b_d \right) \vec{P}_{\kappa} \vec{A}_{d} \vec{P}^{'}_{\kappa}  \right) \leq \phi \left( \vec{A}_{d} \right).
\end{align*}
Hence using the definition of universal optimality, we get that $d^*$ is a universally optimal design.
\end{proof}

\begin{lemma}
Let $d^* \in \varOmega_{t,n=\lambda t \left(t-1 \right),p=t}$ be a design given by $OA_{I} \left( n=\lambda t \left(t-1 \right), p=t, t, 2 \right)$, where $\lambda$ is a positive integer and $t \geq 3$. Then for $g > 1$, the information matrix $\cmat_{d^*}$ is not completely symmetric.
\label{ree2c}
\end{lemma}
\begin{proof}
\label{proof4}
Here $d^* \in \varOmega_{t,n=\lambda t \left(t-1 \right),p=t}$ is a design given by $OA_{I} \left( n=\lambda t \left(t-1 \right), p=t, t, 2 \right)$, where $\lambda$ is a positive integer and $t \geq 3$. From Theorem~{\upshape\ref{thm4b}}, for $g > 1$, the information matrix $\vec{C}_{d^*}$ is given as
\begin{align}
\begin{split}
\vec{C}_{d^*} &= \identity_g \otimes \left[\vec{C}_{{d^*}11} - \vec{C}_{{d^*}12} \vec{C}^{-}_{{d^*}22} \vec{C}_{{d^*}21} \right],
\end{split}
\label{eqq1b}
\end{align}
where $\vec{C}^{-}_{d^*22}$ is a generalized inverse of of $\vec{C}_{d^*22}$.
From \cite{Martin1998Variance-balancedObservations} and \cite{Bose2009OptimalDesigns} (see Chapter 1, pp. 12--18), we get that for $g=1$,
\begin{align}
\begin{split}
\vec{C}_{d^*} &=\vec{C}_{{d^*}11} - \vec{C}_{{d^*}12} \vec{C}^{-}_{{d^*}22} \vec{C}_{{d^*}21} = \left(   \text{\rmfamily\upshape det} \left( \vec{E} \right) / e_{22} \right) \hatmat_t,
\end{split}
\label{eqq2b}
\end{align}
where $\vec{E} = \frac{n}{t-1} \begin{bmatrix}
e_{11} & e_{12}\\
e_{12} & e_{22}
\end{bmatrix}$, $\text{\rmfamily\upshape det} \left( \vec{E} \right) \neq 0$, $e_{11} = \text{\rmfamily\upshape tr} \left(  \vec{T}^{'}_{d^*1} \vec{V}^{*} \vec{T}_{d^*1}\right)$, $e_{12} = \text{\rmfamily\upshape tr} \left(  \vec{T}^{'}_{d^*1} \vec{V}^{*} \vec{\psi} \vec{T}_{d^*1}\right)$ and $e_{22} = \text{\rmfamily\upshape tr} \left(  \vec{T}^{'}_{d^*1} \vec{\psi}^{'} \vec{V}^{*} \vec{\psi} \vec{T}_{d^*1}\right) -  \frac{\left(\vec{V}^{*} \right)_{1,1}}{t}$. Here $\hatmat_t = \identity_t - \frac{1}{t} \vecone_{t} \vecone_{t}^{'}$, $\vec{\psi} =  
\begin{bmatrix}
\zero^{'}_{p-1 \times 1} & 0\\
\identity_{p-1} & \zero_{p-1 \times 1}
\end{bmatrix}$, $\vec{V}^* = \vec{V}^{-1} - \left( \vecone_p^{'} \vec{V}^{-1} \vecone_p \right)^{-1} \vec{V}^{-1} \vecone_{p} \vecone_{p}^{'} \vec{V}^{-1}$, and $\left(\vec{V}^{*} \right)_{1,1}$ is the element corresponding to first row and first column of the matrix $\vec{V}^{*}$.
 
So using \eqref{eqq1b} and \eqref{eqq2b}, it is clear that for $g>1$, the matrix $\cmat_{d^*}$ is completely symmetric if and only if all off-diagonal elements of $\left( \text{\rmfamily\upshape det} \left(\vec{E} \right) / e_{22} \right) \hatmat_t$ are $0$. We know that all off-diagonal elements of the matrix $\hatmat_{t} = \identity_t - \frac{1}{t} \vecone_{t} \vecone_{t}^{'}$ are nonzero. From \eqref{eqq2b}, for $t \geq 3$, we get that $\text{\rmfamily\upshape det} \left( \vec{E}  \right) \neq 0$. Hence for $g>1$ and $t\geq3$, all off-diagonal elements of $\cmat_{d^*}$ are nonzero. Thus for $g>1$ and $t \geq 3$, the information matrix $\cmat_{d^*}$ is not completely symmetric.
\end{proof}

\section{Proofs of Section~{\upshape\ref{information matrices}}}\label{secB}

\begin{proof}[Proof of Theorem~{\upshape\ref{thm4b}}]
\label{proof1}
Let $\vec{\xi}_k = \begin{bmatrix}
\mu_k &
\vec{\alpha}^{'}_k &
\vec{\beta}^{'}_k & 
\vec{\rho}^{'}_k
\end{bmatrix}^{'}
$. Then by rearranging the parameters, the model \eqref{unio} can be equivalently expressed as
\begin{multline}
\begin{bmatrix}
 \vec{Y}^{'}_{d1} &
\cdots &
\vec{Y}^{'}_{dg}
\end{bmatrix}^{'}
= \left(  \identity_g \otimes \vec{T}_d \right) 
\begin{bmatrix}
\vec{\tau}^{'}_1 & \cdots & \vec{\tau}^{'}_g
\end{bmatrix}^{'}
+
\left(
\identity_g \otimes 
\begin{bmatrix}
\vecone_{np} & \vec{X}_1 & \vec{F}_d
\end{bmatrix}
\right) 
\begin{bmatrix}
\vec{\xi}_1^{'} & \cdots & \vec{\xi}_g^{'}
\end{bmatrix}^{'}
\\+
\begin{bmatrix}
 \vec{\eps}^{'}_{1} &
\cdots &
\vec{\eps}^{'}_{g}
\end{bmatrix}^{'}.
\label{sss1}
\end{multline}
Premultiplying the above equation by $\identity_g \otimes \vec{\varSigma}^{-1/2}$, we get the model as
\begin{multline}
\begin{bmatrix}
 \vec{Y}^{'}_{d1(new)} &
\cdots &
\vec{Y}^{'}_{dg(new)}
\end{bmatrix}^{'}
= \left(  \identity_g \otimes \vec{\varSigma}^{-1/2} \vec{T}_d \right) 
\begin{bmatrix}
\vec{\tau}^{'}_1 & \cdots & \vec{\tau}^{'}_g
\end{bmatrix}^{'}
\\+
\left(
\identity_g \otimes \vec{\varSigma}^{-1/2}
\begin{bmatrix}
\vecone_{np} & \vec{X}_1 & \vec{F}_d
\end{bmatrix}
\right) 
\begin{bmatrix}
\vec{\xi}_1^{'} & \cdots & \vec{\xi}_g^{'}
\end{bmatrix}^{'}
+
\begin{bmatrix}
 \vec{\eps}^{'}_{1(new)} &
\cdots &
\vec{\eps}^{'}_{g(new)}
\end{bmatrix}^{'},
\label{sss2}
\end{multline}
where $\vec{Y}_{dk(new)} = \vec{\varSigma}^{-1/2} \vec{Y}_{dk}$ and $\vec{\eps}_{k(new)} = \vec{\varSigma}^{-1/2} \vec{\eps}_{k} $. Note that the dispersion matrix of the vector of transformed error terms $\begin{bmatrix}
 \vec{\eps}^{'}_{1(new)} &
\cdots &
\vec{\eps}^{'}_{g(new)}
\end{bmatrix}^{'}$ is $\sigma^2 \identity_{gnp}$. Thus using the expression of the information matrix from \cite{kunert1983optimal1}, we get that the information matrix for the direct effects can be expressed as
\begin{align}
\vec{C}_{d} &=
\left( \identity_g \otimes \vec{T}_d^{'} \vec{\varSigma}^{-1/2} \right)  \text{\rmfamily\upshape pr}^{\perp} \left( 
\identity_g \otimes \vec{\varSigma}^{-1/2}
\begin{bmatrix}
\vecone_{np} & \vec{X}_1 & \vec{F}_d
\end{bmatrix}  
\right) \left( \identity_g \otimes \vec{\varSigma}^{-1/2} \vec{T}_d\right),
\label{sss3}
\end{align}
where $\text{\rmfamily\upshape pr}^{\perp} \left( \vec{M} \right) = \identity - \vec{M} \left(\vec{M}^{'} \vec{M} \right)^{-} \vec{M}^{'}$ is the orthogonal projection matrix onto the space orthogonal to the column space of matrix $\vec{M}$ and the order of $\identity$ is same as the order of $\vec{M} \left(\vec{M}^{'} \vec{M} \right)^{-} \vec{M}^{'}$. 
By calculation, the above equation can be expressed as
\begin{align}
\vec{C}_{d} &= \identity_g \otimes \left[
 \vec{T}_d^{'} \vec{\varSigma}^{-1/2} \text{\rmfamily\upshape pr}^{\perp} \left( \vec{\varSigma}^{-1/2}
\begin{bmatrix}
\vecone_{np} & \vec{X}_1 & \vec{F}_d
\end{bmatrix}  
\right) \vec{\varSigma}^{-1/2} \vec{T}_d \right].
\label{sss4}
\end{align}
From \cite{Bose2009OptimalDesigns} (see Chapter 1, pp. 12--18), we know
\begin{align}
\vec{T}_d^{'} \vec{\varSigma}^{-1/2} \text{\rmfamily\upshape pr}^{\perp} \left( \vec{\varSigma}^{-1/2}
\begin{bmatrix}
\vecone_{np} & \vec{X}_1 & \vec{F}_d
\end{bmatrix}  
\right) \vec{\varSigma}^{-1/2} \vec{T}_d =  \vec{C}_{{d}11} - \vec{C}_{{d}12} \vec{C}^{-}_{{d}22} \vec{C}_{{d}21},
\label{sss5}
\end{align}
where $\vec{C}_{{d}11} = \vec{T}^{'}_d \vec{A}^* \vec{T}_d$, $\vec{C}_{{d}12} = \vec{C}_{{d}21}^{'} = \vec{T}^{'}_d \vec{A}^* \vec{F}_d$, $\vec{C}_{{d}22} = \vec{F}^{'}_d \vec{A}^* \vec{F}_d$, $\vec{A}^* = \vec{\varSigma}^{-1/2} \text{\rmfamily\upshape pr}^{\perp} \left( \vec{\varSigma}^{-1/2} \vec{X}_1 \right) \vec{\varSigma}^{-1/2} = \hatmat_n \otimes \vec{V}^*$ and $\vec{V}^* = \vec{V}^{-1} - \left( \vecone_p^{'} \vec{V}^{-1} \vecone_p \right)^{-1} \vec{V}^{-1} \vecone_{p} \vecone_{p}^{'} \vec{V}^{-1}$. Here $\hatmat_n = \identity_n - \frac{1}{n} \vecone_n \vecone_n^{'}$ and $\vec{C}^{-}_{d22}$ is a generalized inverse of of $\vec{C}_{d22}$. Thus from \eqref{sss4} and \eqref{sss5}, we can prove \eqref{p9}.
\end{proof}

\begin{proof}[Proof of Remark~{\upshape\ref{re2a}}]
\label{proof2}
From \cite{Bose2009OptimalDesigns} (see Chapter 1, pp. 12--18), we get that for $g=1$, under model \eqref{unio}, the information matrix for the direct effects, which is given as $ \vec{C}_{{d}11} - \vec{C}_{{d}12} \vec{C}^{-}_{{d}22} \vec{C}_{{d}21}$, is symmetric, n.n.d., have row sums and column sums as zero, and is invariant with respect to the choice of generalized inverses involved.\\
Thus using the expression of the information matrix given in Theorem~{\upshape\ref{thm4b}}, we get that for $g \geq 1$, $\vec{C}_{d}$ is symmetric, n.n.d., satisfies
\begin{align*}
\vec{C}_{d} \vecone_{gt} &= \left( \identity_g \otimes \left[ \vec{C}_{{d}11} - \vec{C}_{{d}12} \vec{C}^{-}_{{d}22} \vec{C}_{{d}21} \right]  \right)\vecone_{gt} = \zero_{gt \times 1},\\
\vecone^{'}_{gt}\vec{C}_{d}  &= \vecone^{'}_{gt} \left( \identity_g \otimes \left[ \vec{C}_{{d}11} - \vec{C}_{{d}12} \vec{C}^{-}_{{d}22} \vec{C}_{{d}21} \right]  \right) = \zero_{1 \times gt}
\end{align*}
and is invariant with respect to the choice of generalized inverses involved.
\end{proof}

\end{appendices}


\makeatletter
\providecommand{\doi}[1]{%
  \begingroup
    \let\bibinfo\@secondoftwo
    \urlstyle{rm}%
    \href{http://dx.doi.org/#1}{%
      doi:\discretionary{}{}{}%
      \nolinkurl{#1}%
    }%
  \endgroup
}
\makeatother

\renewcommand{\doi}[1]{\urlstyle{rm}\url{https://doi.org/#1}}


\begin{thebibliography}{}

\bibitem[Bose and Dey, 2009]{Bose2009OptimalDesigns}
Bose, M. and Dey, A. (2009).
\newblock {\em Optimal crossover designs}.
\newblock World Scientific, Hackensack, NJ. \doi{10.1142/9789812818430}.

\bibitem[Bose and Dey, 2013]{Bose2013DevelopmentsDesigns}
Bose, M. and Dey, A. (2013).
\newblock Developments in crossover designs.

\bibitem[Carriere and Huang, 2000]{Carriere2000CrossoverTrials}
Carriere, K.~C. and Huang, R. (2000).
\newblock Crossover designs for two-treatment clinical trials.
\newblock {\em Journal of Statistical Planning and Inference}, 87(1):125--134. \doi{10.1016/S0378-3758(99)00185-8}.

\bibitem[Carriere and Reinsel, 1993]{Carriere1993OptimalTreatments}
Carriere, K.~C. and Reinsel, G.~C. (1993).
\newblock Optimal two-period repeated measurement designs with two or more
  treatments.
\newblock {\em Biometrika}, 80(4):924--929. \doi{10.1093/biomet/80.4.924}.

\bibitem[Cheng and Wu, 1980]{Cheng1980BalancedDesigns}
Cheng, C.~S. and Wu, C.~F. (1980).
\newblock Balanced repeated measurements designs.
\newblock {\em The Annals of Statistics}, 8(6):1272--1283. \doi{10.1214/aos/1176345200}.

\bibitem[Cochran, 1939]{Cochran1939Long-termExperiments}
Cochran, W.~G. (1939).
\newblock Long-term agricultural experiments.
\newblock {\em Supplement to the Journal of the Royal Statistical Society},
  6(2):104--148. \doi{10.2307/2983686}.

\bibitem[Hedayat and Afsarinejad, 1975]{Hedayat1975RepeatedI}
Hedayat, A.~S. and Afsarinejad, K. (1975).
\newblock Repeated measurements design, I.
\newblock In {\em A Survey of Statistical Design and Linear Models} (J. N.
  Srivastava Ed.), pp. 229--242. North Holland, Amsterdam.

\bibitem[Hedayat and Afsarinejad, 1978]{Hedayat1978RepeatedII}
Hedayat, A.~S. and Afsarinejad, K. (1978).
\newblock Repeated measurements designs, II.
\newblock {\em The Annals of Statistics}, 6(3):619--628. \doi{10.1214/aos/1176344206}.

\bibitem[Hedayat et~al., 2006]{Hedayat2006OptimalRandom}
Hedayat, A.~S., Stufken, J., and Yang, M. (2006).
\newblock Optimal and efficient crossover designs when subject effects are
  random.
\newblock {\em Journal of the American Statistical Association},
  101(475):1031--1038. \doi{10.1198/016214505000001384}.

\bibitem[Hedayat and Yang, 2003]{Hedayat2003UniversalDesigns}
Hedayat, A.~S. and Yang, M. (2003).
\newblock Universal optimality of balanced uniform crossover designs.
\newblock {\em The Annals of Statistics}, 31(3):978--983. \doi{10.1214/aos/1056562469}.

\bibitem[Hedayat and Yang, 2004]{Hedayat2004UniversalDesigns}
Hedayat, A.~S. and Yang, M. (2004).
\newblock Universal optimality for selected crossover designs.
\newblock {\em Journal of the American Statistical Association},
  99(466):461--466. \doi{10.1198/016214504000000331}.

\bibitem[Jankar et~al., 2020]{Jankar2020OptimalModels}
Jankar, J., Mandal, A., and Yang, J. (2020).
\newblock Optimal crossover designs for generalized linear models.
\newblock {\em Journal of Statistical Theory and Practice}, 14(2):1--27. \doi{10.1007/s42519-020-00089-5}.

\bibitem[Kenward and Jones, 2014]{Kenward2014CrossoverTrials}
Kenward, M.~G. and Jones, B. (2014).
\newblock Crossover trials.
\newblock In {\em Methods and applications of statistics in clinical trials}, (Vol. 1), pp. 310--319. Wiley,
  Hoboken, NJ. \doi{10.1002/9781118596005.ch27}.

\bibitem[Kiefer, 1975]{Kiefer1975ConstructionIi}
Kiefer, J. (1975).
\newblock Construction and optimality of generalized {Y}ouden designs.
\newblock In {\em A survey of statistical design and linear models} (J. N.
  Srivastava Ed.), pp. 333--353. North-Holland, Amsterdam-Oxford.

\bibitem[Kunert, 1983]{kunert1983optimal1}
Kunert, J. (1983).
\newblock Optimal design and refinement of the linear model with applications
  to repeated measurements designs.
\newblock {\em The Annals of Statistics}, 11(1):247--257. \doi{10.1214/aos/1176346075}.

\bibitem[Kunert, 1984]{Kunert1984OptimalityDesigns}
Kunert, J. (1984).
\newblock Optimality of balanced uniform repeated measurements designs.
\newblock {\em The Annals of Statistics}, 12(3):1006--1017. \doi{10.1214/aos/1176346717}.

\bibitem[Kunert and Martin, 2000]{Kunert2000OptimalityErrors}
Kunert, J. and Martin, R.~J. (2000).
\newblock Optimality of type I orthogonal arrays for cross-over models with
  correlated errors.
\newblock {\em Journal of Statistical Planning and Inference}, 87(1):119--124. \doi{10.1016/S0378-3758(99)00182-2}.

\bibitem[Kushner, 1997]{Kushner1997OptimalityDesigns}
Kushner, H.~B. (1997).
\newblock Optimality and efficiency of two-treatment repeated measurements
  designs.
\newblock {\em Biometrika}, 84(2):455--468. \doi{10.1093/biomet/84.2.455}.

\bibitem[Kushner, 1998]{Kushner1998OptimalObservations}
Kushner, H.~B. (1998).
\newblock Optimal and efficient repeated-measurements designs for uncorrelated
  observations.
\newblock {\em Journal of the American Statistical Association},
  93(443):1176--1187. \doi{10.1080/01621459.1998.10473778}.

\bibitem[Laska and Meisner, 1985]{Laska1985AModels}
Laska, E.~M. and Meisner, M. (1985).
\newblock A variational approach to optimal two-treatment crossover designs:
  application to carryover-effect models.
\newblock {\em Journal of the American Statistical Association},
  80(391):704--710. \doi{10.1080/01621459.1985.10478172}.

\bibitem[Laska et~al., 1983]{Laska1983OptimalEffects}
Laska, E.~M., Meisner, M., and Kushner, H.~B. (1983).
\newblock Optimal crossover designs in the presence of carryover effects.
\newblock {\em Biometrics}, 39(4):1087--1091. \doi{10.2307/2531343}.

\bibitem[Leaker et~al., 2017]{Leaker2017TheComplement}
Leaker, B.~R., Malkov, V.~A., Mogg, R., Ruddy, M.~K., Nicholson, G.~C., Tan,
  A.~J., Tribouley, C., Chen, G., De~Lepeleire, I., Calder, N.~A., Chung, H., Lavender, P., Carayannopoulos, L. N., and Hansel, T. T. (2017).
\newblock The nasal mucosal late allergic reaction to grass pollen involves
  type 2 inflammation ({IL}-5 and {IL}-13), the inflammasome ({IL}-1{$\beta$}), and
  complement.
\newblock {\em Mucosal immunology}, 10(2):408--420. \doi{10.1038/mi.2016.74}.

\bibitem[Martin and Eccleston, 1998]{Martin1998Variance-balancedObservations}
Martin, R.~J. and Eccleston, J.~A. (1998).
\newblock Variance-balanced change-over designs for dependent observations.
\newblock {\em Biometrika}, 85(4):883--892. \doi{10.1093/biomet/85.4.883}.

\bibitem[Mukhopadhyay et~al., 2021]{Mukhopadhyay2021LocallyDesigns}
Mukhopadhyay, S., Singh, S.~P., and Singh, A. (2021).
\newblock Locally optimal binary crossover designs.
\newblock {\em Stat Appl}, 19(1):223--246.

\bibitem[Pareek et~al., 2023]{pareek2023likelihood}
Pareek, S., Das, K., and Mukhopadhyay, S. (2023).
\newblock Likelihood-based missing data analysis in crossover trials.
\newblock {\em Brazilian Journal of Probability and Statistics},
  37(2):329--350. \doi{10.1214/23-BJPS570}.

\bibitem[Putt and Chinchilli, 1999]{Putt1999AStudies}
Putt, M. and Chinchilli, V.~M. (1999).
\newblock A mixed effects model for the analysis of repeated measures
  cross-over studies.
\newblock {\em Statistics in Medicine}, 18(22):3037--3058. \doi{10.1002/(SICI)1097-0258(19991130)18:22<3037::AID-SIM243>3.0.CO;2-7}.

\bibitem[Senn, 2002]{Senn2002Cross-overResearch}
Senn, S.~S. (2002).
\newblock {\em Cross-over trials in clinical research}, (Vol. 5).
\newblock John Wiley {\&} Sons. \doi{10.1002/0470854596}.

\bibitem[Singh and Kunert, 2021]{Singh2021EfficientSettings}
Singh, R. and Kunert, J. (2021).
\newblock Efficient crossover designs for non-regular settings.
\newblock {\em Metrika}, 84(4):497--510. \doi{10.1007/s00184-020-00780-4}.

\bibitem[Singh and Mukhopadhyay, 2016]{Singh2016BayesianModels}
Singh, S.~P. and Mukhopadhyay, S. (2016).
\newblock Bayesian crossover designs for generalized linear models.
\newblock {\em Computational Statistics {\&} Data Analysis}, 104:35--50. \doi{10.1016/j.csda.2016.06.002}.

\bibitem[Singh et~al., 2021]{Singh2021MinmaxTrials}
Singh, S.~P., Mukhopadhyay, S., and Raj, H. (2021).
\newblock Min–max crossover designs for two treatments binary and poisson
  crossover trials.
\newblock {\em Statistics and Computing}, 31(5):1--11. \doi{10.1007/s11222-021-10029-3}.

\bibitem[Stufken, 1991]{Stufken1991SomeDesigns}
Stufken, J. (1991).
\newblock Some families of optimal and efficient repeated measurements designs.
\newblock {\em Journal of Statistical Planning and Inference}, 27(1):75--83. \doi{10.1016/0378-3758(91)90083-Q}.

\bibitem[Yeh, 1986]{yen1986conditions}
Yeh, C.~M. (1986).
\newblock Conditions for universal optimality of block designs.
\newblock {\em Biometrika}, 73(3):701--706. \doi{10.1093/biomet/73.3.701}.

\end{thebibliography}
\end{document}